\documentclass[10pt,letter,runningheads]{llncs}
\usepackage[T1]{fontenc}   
\usepackage{amsmath}
\usepackage{amssymb}
\usepackage{graphicx}
\usepackage{color}

\definecolor{myred1}{rgb}{1.0,0.0,0.1}
\definecolor{myred2}{rgb}{1.0,0.0,0.3}
\definecolor{myred3}{rgb}{1.0,0.0,0.5}

\definecolor{myblue1}{rgb}{0.3,0.0,1.0}
\definecolor{myblue2}{rgb}{0.0,0.3,1.0}
\definecolor{myblue3}{rgb}{0.1,0.1,1.0}

\definecolor{mygreen1}{rgb}{0.2,1.0,0.0}
\definecolor{mygreen2}{rgb}{0.0,1.0,0.2}
\definecolor{mygreen3}{rgb}{0.0,1.0,0.5}

\newtheorem{fact}{Fact}
\newcommand{\junk}[1]{}

\begin{document}




\title{Fast approximate $\ell$-center clustering in high dimensional spaces}
\author{
Miros{\l}aw Kowaluk
\inst{1}
\and
Andrzej Lingas
\inst{2}
\and
  Mia Persson
  \inst{3}}
\institute{
  Institute of Informatics, University of Warsaw, Warsaw, Poland.
  \texttt{kowaluk@mimuw.edu.pl}
\and
  Department of Computer Science, Lund University,
Lund, Sweden.
\texttt{Andrzej.Lingas@cs.lth.se}
\and
Department of Computer Science and Media Technology, Malm\"o University, Malm\"o, Sweden.
\texttt{Mia.Persson@mau.se}
}
\pagestyle{plain}
\maketitle
\begin{abstract}
  We study the design of efficient approximation algorithms for the
  $\ell$-center clustering and minimum-diameter $\ell$-clustering
  problems in high dimensional Euclidean and Hamming spaces.  Our main
  tool is randomized dimension reduction. First, we present a general
  method of reducing the dependency of the running time of a
  hypothetical algorithm for the $\ell$-center problem in a high
  dimensional Euclidean space on the dimension size.  Utilizing in
  part this method, we provide $(2+\epsilon)$- approximation
  algorithms for the $\ell$-center clustering and minimum-diameter
  $\ell$-clustering problems in Euclidean and Hamming spaces that are
  substantially faster than the known $2$-approximation ones when both
  $\ell$ and the dimension are super-logarithmic. Next, we apply the
  general method to the recent fast approximation algorithms with
  higher approximation guarantees for the $\ell$-center clustering
  problem in a high dimensional Euclidean space. Finally, we provide a
  speed-up of the known $O(1)$-approximation method for the
  generalization of the $\ell$-center clustering problem to include
  $z$ outliers (i.e., $z$ input points can be ignored while computing
  the maximum distance of an input point to a center) in high
  dimensional Euclidean and Hamming spaces.
  \end{abstract}

\begin{keywords}
 $\ell_2$ distance, Euclidean space, Hamming distance, Hamming space, clustering, approximation algorithm, time complexity
\end{keywords}
\section{Introduction}
Clustering is nowadays a standard tool in the data
  analysis in computational biology/medical sciences, computer vision,
  and machine learning.
One of the most popular variants
  of clustering are the $\ell$-center clustering problem and the related
  minimum-diameter $\ell$-clustering problem in metric spaces
  (among other things in Euclidean and Hamming spaces).
  Given a finite set $P$ of
points in a metric space,
the first problem asks for 
finding a set of $\ell$ points in the metric space, called
{\em centers}, such that the maximum distance of a point
in $P$ to its nearest center is minimized.
The second problem asks for partitioning
the input point set $P$ into $\ell$ clusters such that the maximum
of cluster diameters is minimized.
  They are known to be NP-hard and even
  NP-hard to approximate within $2-\epsilon$ for any
  constant $\epsilon > 0$ \cite{Gon85}.
  Gonzalez provided a simple $2$-approximation method
 for $\ell$-center clustering that yields also a $2$-approximation
 for minimum-diameter $\ell$-clustering \cite{Gon85}.
 In case of a  $d$-dimensional space, his method takes
 $O(nd\ell)$ time, where $n$ is the number of input points and $d$
 stands for the dimension size.
 For Euclidean spaces of bounded
 dimension, and more generally, for metric spaces of bounded doubling
 dimension, there exist faster $2$-approximation algorithms for the
 $\ell$-center problem with hidden exponential dependence on the
 dimension in their running time, see \cite{FG88} and \cite{HM06},
 respectively.
 There are also several newer works on speeding up $\ell$-center
 approximation algorithms in Euclidean spaces 
 by worsening the approximation guarantee
 (e.g., \cite{EHS20,FJLNP25,JKS24}, and for a bicriteria variant
 \cite{EFHSW24}).
 In particular, trade-offs between approximation guarantees
 of the form $O(\alpha)$ and the running times
$poly(d\log n)
  (n+\ell^{1+1/\alpha^2}n^{O(1/\alpha^{2/3})})$ and
  $poly(d\log n)n\ell^{1/\alpha^2}$ have been obtained in
  \cite{FJLNP25}\footnote{ In \cite{FJLNP25},
    the notation $\tilde{O}(\ )$ suppresses $poly(d \log n)$ factors.}
  by refining an earlier result in \cite{EHS20}.
 In some of the aforementioned applications of both
 $\ell$-clustering problems,
 massive datasets in a high dimensional metric space combined
 with a large value of the parameter $\ell$ may occur. In such situations,
 neither the $O(nd\ell)$-time method
 nor the ones of substantially worse
 approximation guarantees or time complexity
 heavily dependent on $d$ are sufficiently useful.
\subsection{Our contributions}
We focus on the design of efficient approximation algorithms for the
  $\ell$-center clustering and minimum-diameter $\ell$-clustering
  problems in high dimensional Euclidean and Hamming spaces.
  First, we
  present a general method of reducing the dependency of the running time
  of a hypothetical algorithm for the $\ell$-center problem in a
  high dimensional Euclidean space on the dimension size.
  The algorithm is required to be {\em conservative}, i.e., to return
  always centers belonging to the input point set.
  The method relies on randomized dimension reduction and
  almost preserves the approximation guarantee of the algorithm.
  Utilizing in part this method, we provide $(2+\epsilon)$-
  approximation algorithms for the $\ell$-center clustering and
  minimum-diameter $\ell$-clustering problems in Euclidean and Hamming
  spaces that are substantially faster than the known
  $2$-approximation ones when both $\ell$ and the dimension are
  super-logarithmic. Next, we apply the general method to the
  recent fast approximation algorithms with higher approximation
  guarantees for the $\ell$-center clustering problem in a
  high dimensional Euclidean space.
   Finally, we provide a speed-up
  of the known $O(1)$-approximation method for the generalization
  of the $\ell$-center clustering problem to include $z$ outliers
  (i.e., $z$ input points can be ignored while computing
  the maximum distance of an
  input point to a center) in high dimensional Euclidean
  and Hamming spaces \cite{Char01}.
  The speed-up is also based
  on randomized dimension reduction and the resulting approximation guarantee
  is only slightly larger than the original one. See also Table~1 for a summary of our contributions.
  \junk{We show that our lemma on randomized dimension reduction
in Hamming spaces yields also $(2+\epsilon)$- approximation algorithms
for the $\ell$-center clustering
and minimum-diameter $\ell$-clustering problems in a Hamming space
$\{0,1\}^d$ that are substantially faster than the known $2$-approximation
ones when both $\ell$ and $d$ are super-logarithmic.}

\subsection{Techniques}
Our fast randomized algorithms for approximate $\ell$-center clustering
and minimum-diameter $\ell$-clustering problems in high dimensional
Euclidean and Hamming spaces
are based on a  variant
of randomized
dimension reduction in Euclidean spaces
given by Achlioptas in \cite{achi03}
and the observation that the Hamming distance
between two 0-1 vectors is equal to their squared
$\ell_2$ distance.
The main idea of a randomized dimension reduction
  is to provide a uniform random map from a $d$-dimensional metric space
  to its $k$-dimensional subspace that preserves distances up
  to $1\pm\epsilon$ factor, where $k=O(\log n),$ with high
  probability. Johnson and Lindenstrauss were first to provide
  such maps from $\mathbb{R}^d$ to $\mathbb{R}^k$ for the $\ell_2$ norm \cite{JL84}.
  The advantage of Achlioptas' variant of JL dimension reduction
  is that such a random map can be generated by binary coins in
  this variant \cite{achi03}.
  \junk{
  Brinkman and Charikar showed that analogous maps cannot exist
  for the $\ell_1$-norm
  \cite{BC05}.
  Therefore, Kushilevitz et al.
 \cite{KOR00} introduced
 a weaker concept of a randomized reduction
 in Hamming spaces that works only
  for a fixed scale (i.e., for pairwise distances below a fixed threshold
  in $[1,d]$, see
  Lemma 2.3 in \cite{AIR18} and Lemma \ref{lem: R} in this paper).}

\subsection {Paper organization}
 The next section contains basic
 definitions and facts on the randomized dimension reduction.
 Section 3 presents our general method for reducing the dependency on
 the dimension size. Section 4 provides our fast
 $(2+\epsilon)$-approximation algorithms for the $\ell$-center clustering
 and minimum-diameter $\ell$-clustering problems in high dimensional
 Euclidean and Hamming spaces.
 Section~5 presents applications of our method from Section 3 to recent
 fast approximation algorithms for $\ell$-center
 clustering in high dimensional Euclidean spaces.
 Section 6 provides a speed-up of the known greedy $O(1)$-approximation method
 for  the $\ell$-center clustering problem with outliers in
 a high dimensional Euclidean space.
We conclude with final remarks.

\begin{table*}[t]
\begin{center}
\begin{tabular}{||c|c|c|c|c||} \hline \hline
problem & metric  & approx.  & time complexity & reference
\\ \hline \hline
$\ell$-center/min-diam. &  arbitrary &  $2$ & $O(nd\ell)$ & \cite{Gon85}
 \\ \hline
 as above & Euclid.,Ham. & $2+\epsilon$ & $O(n\log n(d + \ell)/\epsilon^2)$ & this paper
\\ \hline \hline
\junk{  $\ell$-center & Euclidean & $O(\alpha)$ & $poly(d\log n)n^{1+1/\alpha^2}$ &
  Eppstein {\em et al.} \cite{EHS20}
  \\ \hline
  as above  & Euclidean & $O(\alpha)$ &
 $\tilde{O}(nd/\epsilon^2+n^{1+1/\alpha^2})$ & this paper
  \\ \hline \hline}
 $\ell$-center & Euclidean & $O(\alpha)$ &  $poly(d\log n)
  (n+\ell^{1+1/\alpha^2}n^{O(1/\alpha^{2/3})})$ & \cite{FJLNP25}$+$\cite{EHS20}
  \\ \hline
  as above  & Euclidean & $O(\alpha)$ &
        $\tilde{O}(nd/\epsilon^2+\ell^{1/\alpha^2}n^{O(1/\alpha^{2/3})})$ & this paper                                             
\\ \hline \hline
  $\ell$-center & Euclidean & $O(\alpha)$ & $poly(d\log n)n\ell^{1/\alpha^2}$ &  \cite{FJLNP25}$+$\cite{EHS20}
  \\ \hline
  as above  & Euclidean & $O(\alpha)$ &
  $\tilde{O}(nd/\epsilon^2+n\ell^{1/\alpha^2})$ & this paper
 \\ \hline \hline             
  $\ell$-center $+$ outliers & arbitrary & $3$ & $poly(n,d,\ell)$ &
  \cite{Char01}
\\ \hline            
  as above & Euclid.,Ham. & $3+\epsilon $ & $\tilde{O}(n^2(\epsilon^{-2}+ \ell))$ & this paper 
\\ \hline \hline
\end{tabular}
\label{table: one}
\vskip 0.5cm
\caption{A summary of speed-ups of approximation algorithms
  for $\ell$-center problems obtained by randomized
  dimension reduction presented in this paper. Importantly, the approximation
  guarantees for our algorithms hold w.h.p.
  The approximation guarantee of $3$ for the $\ell$-center
  problem with outliers in an arbitrary metric space in \cite{Char01}
  as well as our guarantee of $3+\epsilon$ in
  an Euclidean or Hamilton space are derived with respect
  to an optimal conservative solution, where the centers
  belong to input points. The authors of \cite{Char01} claim that they
  can remove this assumption in case of an Euclidean space.
}
\end{center}
\end{table*}

\section{Preliminaries}

For a positive integer $r$, $[r]$ stands for the set of positive
integers not exceeding $r.$ The cardinality of a finite set $S$ is
denoted by $|S|.$

\junk{For two vectors $(a_1,\dots,a_d)$ and $(b_1,\dots,b_d)$ in $\mathbb{R}^d,$
their {\em inner product} equals $\sum_{\ell=1}^d a_{\ell}b_{\ell}.$
}
The transpose of a matrix $D$ is denoted by $D^{\top}.$
\junk{If the entries of $D$ are in $\{0,1\}$ then
$D$ is a 0-1 matrix while if they are in a finite alphabet $\Sigma,$
$D$ is a $\Sigma$ matrix.

The symbol $\omega$ denotes the smallest real number such that two $n\times n$
matrices can be multiplied using $O(n^{\omega +\epsilon})$
operations over the field of reals, for all $\epsilon > 0.$
}

The {\em Hamming distance} between two points $a,\ b$ (vectors) in
$\{0,1\}^d$
is the number of the
coordinates in which the two points differ.  Alternatively,
it can be defined as the distance between $a$ and $b$ in
the $\ell_1$ metric over $\{0,1\}^d.$ It is denoted by $\mathrm{ham}(a,b).$
The distance between two real vectors
$a,\ b$ in the $\ell_2$ metric
is denoted by $||a-b||_2.$
\junk{
For a positive real $\epsilon,$ an estimation of the Hamming distance
between two points $a,\ b\in \{0,1\}^d$ whose value
is in
$[\mathrm{ham}(a,b)/(1+\delta), (1+\delta)\mathrm{ham}(a,b)]$
differs from
$\mathrm{ham}(a,b)$ at most by $\epsilon \cdot \mathrm{ham}(a,b)$
is called an
{\em $\epsilon$-approximation} of $\mathrm{ham}(a,b).$
For a finite set $S$ of points in $\{0,1\}^d$,
an {\em $\epsilon$-approximate nearest neighbor}
of a point $a\in \{0,1\}^d$ in $S$ is a point $b\in S\setminus \{a\}$
such that $\mathrm{ham}(a,b)\le (1+\epsilon)\mathrm{ham}(a,c)$
for any point $c\in S\setminus \{ a\}.$}

An event is said to hold {\em with high probability} (w.h.p. for short) in
terms of a parameter $N$ related to the input size if it holds with
probability at least $1-\frac 1 {N^{\alpha}}$, where $\alpha $ is any
constant not less than $1.$

The following fact and  corollaries enable an efficient randomized dimension
reduction in high dimensional Euclidean and Hamming spaces.
\vfill
\newpage
\begin{fact}\label{fact: achi03}(Achlioptas \cite{achi03})
Let $P$ be an arbitrary set of $n$ points in $\mathbb{R}^d$, represented as an
$n\times d$ matrix $A.$ Given $\epsilon,\ \beta >0,$
let $k_0=\frac {4+2\beta} {\epsilon^2/2 -\epsilon^3/3} \log n.$
For an integer $k\ge k_0,$ let $R$ be a $d\times k$ random matrix
$(R_{ij})$, where $R_{ij}=1$
with probability $\frac 12$ and $R_{ij}=-1$ otherwise.
Let $E=\frac 1{\sqrt k} AR$ and let $f:\mathbb{R}^d \rightarrow \mathbb{R}^k$
map the the $i$-th row of $A$ on the $i$-th row of $E$.
With probability at least $1-n^{-\beta},$ for all $u, v\in P,$

$$(1-\epsilon)(||u-v||_2)^2 \le (||f(u)-f(v)||_2)^2\le (1+\epsilon)
(||u-v||_2)^2.$$
\end{fact}

Since $1-\epsilon\le  \sqrt {1-\epsilon}$
and $\sqrt {1+\epsilon}\le 1+\epsilon$ for $\epsilon \in (0,1),$
we immediately obtain the following
corollary  from Fact \ref{fact: achi03}.

\begin{corollary}\label{cor: newachi}
  Assume the notation from Fact \ref{fact: achi03}.
  Suppose that $\epsilon \in (0,1)$ and $P\subset \mathbb{R}^d.$ Then, w.h.p.
for all $u, v \in P$,
$$(1-\epsilon) ||u-v||_2 \le  ||f(u)-f(v)||_2\le (1+\epsilon)
||u-v||_2.$$
\end{corollary}

Observe that if $u, v \in \{0,1\}^d$ then $\mathrm{ham}(u,v)=
(||u-v||_2)^2.$ Hence, we obtain also the following
corollary  from Fact \ref{fact: achi03}.

\begin{corollary}\label{cor: achi03}
  Assume the notation from Fact \ref{fact: achi03}.
  Suppose that $\epsilon > 0$
  and $P\subset \{0,1\}^d\subset \mathbb{R}^d.$ Then, w.h.p.
for all $u, v \in P$,
$$(1-\epsilon) \mathrm{ham}(u,v)\le (||f(u)-f(v)||_2)^2\le (1+\epsilon)
\mathrm{ham}(u,v).$$
\end{corollary}

\section{Approximate $\ell$-center clustering in high dimensional spaces}
The $\ell$-center clustering problem  
 in the Euclidean $\mathbb{R}^d$ is as follows: given a set $P$ of $n$ points
 in  $\mathbb{R}^d$, find a set $T$ of $\ell$ points in  $\mathbb{R}^d$
 that minimize \\$\max_{v\in P} \min_{u\in T} ||v-u||_2.$

The minimum-diameter $\ell$-clustering problem in 
 the Euclidean $\mathbb{R}^d$ is as follows: given a set $P$ of $n$ points
 in  $\mathbb{R}^d$, find a partition of $P$ into $\ell$ subsets
 $P_1,P_2,\dots, P_{\ell}$ that minimize
 $\max_{i\in [\ell]} \max_{v,u \in P_i} ||v-u|||_2.$

 The $\ell$-center clustering problem could be
 also termed as the minimum-radius $\ell$-clustering problem.
 Both problems are known to be NP-hard
 to approximate within $2-\epsilon$ 
 in metric spaces \cite{Gon85,HS85}.
 \junk{Gonzalez provided a simple $2$-approximation method
 for $\ell$-center clustering that yields also a $2$-approximation
 for minimum-diameter $\ell$-clustering \cite{Gon85}.}

Gonzalez' simple $2$-approximation method for
$\ell$-center clustering yields also a $2$-approximation
 for minimum-diameter $\ell$-clustering \cite{Gon85}.
 It picks an arbitrary input point as the first
 center and repeatedly extends the current center set 
 by an input point that maximizes the distance
 to the current center set, until $\ell$ centers are found.
 In case of $d$-dimensional Euclidean or Hamming
 space, his method takes
 $O(nd\ell)$ time, where $n$ is the number of input points.
\junk{
 On a word RAM with computer words of length $w$, the factor $w$ can
 be shaved off. For low dimensional Euclidean spaces of bounded
 dimension, and more generally, for metric spaces of bounded doubling
 dimension, there exist faster $2$-approximation algorithms for the
 $\ell$-center problem with hidden exponential dependence on the
 dimension in their running time, see \cite{FG88} and \cite{HM06},
 respectively.}
   
 By forming for each of the $\ell$ centers, the cluster
consisting of all input points for which this center
is the closest one (with ties solved arbitrarily),
one obtains an $\ell$-clustering with the maximum
cluster diameter within two of the minimum \cite{Gon85}.
\junk{
In this section, we provide a faster randomized $(2+\epsilon)$-approximation
method for both problems in in the Euclidean $\mathbb{R}^d,$ 
when $d=\omega (\log n)$ and
$\ell=\omega (\log n)$. It combines Gonzalez' method
with the randomized dimension reduction described in
Fact \ref{fact: achi03}.}

We shall call an approximation algorithm for the $\ell$-center
clustering problem {\em conservative} if it always returns centers
belonging to the input point set.

In this section, we present a general randomized method of decreasing the
dependency of the running time of a hypothetical conservative
approximation algorithm for the $\ell$-center clustering problem in
the Euclidean $\mathbb{R}^d$ on $d,$ by using the randomized dimension
reduction given in Fact \ref{fact: achi03}.
\vfill
\newpage

\bigskip
\noindent\fbox{%
    \parbox{\textwidth}{%
    \noindent
        {\bf procedure} $DIMREDCENTER(\ell, P, \epsilon , SR)$
        \par
        \noindent
            {\em Input}: A positive integer $\ell,$
a set $P$ of points $p_1,\dots,p_n \in \mathbb{R}^d,$ $n>\ell,$
a real $\epsilon\in (0,\frac 12),$ and a conservative approximation
subroutine $SR$
for the $m$-center clustering problem in $\mathbb{R}^q,$
where $m\le \ell$ and $q\le d.$
            \par
            \noindent
        {\em Output}: A set $T$ of $\ell$ centers of $P.$

                \begin{enumerate}
\item
                  Set $n$ to the number of input points
                 and $k$ to $O(\log n /\epsilon^2)$.
                 \item
                 Generate a random $k\times d$     matrix $R$ with entries
                 in $\{-1,1\}$,
                defining the function $f:\mathbb{R}^d \rightarrow 
                \mathbb{R}^k$ by $f(x)=Rx^{\top}$ 
                (see Fact \ref{fact: achi03}).
                \item
                  Compute the values of the function $f$  for each
                  point $p_i\in P,$
                  i.e.,
                  for $i=1,\dots,n$,
                  compute $Rp_i^{\top}.$
                  Also, for $i=1,\dots,n,$ if the value of $f^{-1}$
                  is not yet defined on $f(p_i)$ then set it to $p_i.$
\item
  Set $\ell'=\min\{\ell, |f(P)|\}$
  and compute a set $T'$ of
  $\ell'$ centers of $f(P)=\{ f(p_1),\dots,f(p_n)\}$ by running
the subroutine $SR$ for the $\ell'$-center clustering problem
  on $f(P).$
\item
  Set $T$ to $\{ f^{-1}(u')| u'\in T'\}$, if $\ell'< \ell$
  then extend $T$ by $\ell-\ell'$ arbitrary points in $P\setminus T,$
  and return it.
\end{enumerate}
}%
}

\bigskip
By Corollary \ref{cor: newachi}, we immediately obtain the following lemma.

                \begin{lemma}\label{lem: ineq1}
                  Assume the notation from Fact \ref{fact: achi03}.
                  Let $T$ be a set of $\ell$ centers of the input
                  point set $P.$
                  For any $\epsilon > 0,$ 
                  the following inequalities hold w.h.p.:
                 $$(1-\epsilon)\max_{v\in P} \min_{u\in T} ||v-u||_2\le
                 \max_{v\in P} \min_{u\in T} ||f(v)-f(u)||_2,$$
                 $$\max_{v\in P} \min_{u\in T} ||f(v)-f(u)||_2\le
                 (1+\epsilon)\max_{v\in P} \min_{u\in T} ||v-u||_2.$$
                  \end{lemma}

                  \begin{lemma}\label{lem: ineq2}
                  Assume the notation from Fact \ref{fact: achi03}.
                  For any $\epsilon \in (0,1/2),$ w.h.p. for all $v,u\in P,$
                  the following inequalities hold:
$$(1-\epsilon) ||f(v)-f(u)||_2 \le ||v-u||_2,$$ 
$$||v-u||_2\le (1+2\epsilon)||f(v)-f(u)||_2.$$
               \end{lemma}
               \begin{proof}

                 By the right-hand
                 inequality in Corollary \ref{cor: newachi}, we
                 obtain for any $v, u \in P,$ and $\epsilon \in
                 (0,1/2),$ $\frac {||f(v)-f(u)||_2}{1+\epsilon} \le
                 ||v-u||_2.$ It follows that $||f(v)-f(u)||_2- \frac
                 {\epsilon}{1+\epsilon} ||f(v)-f(u)||_2 \le
                 ||v-u||_2.$ Consequently, we obtain the first
                 inequality in this lemma.

                 Similarly, by the left-hand inequality
                 in Corollary \ref{cor: newachi}, we
                 infer that for any $v, u \in P,$ and $\epsilon \in
                 (0,1/2),$ $||v-u||_2 \le \frac {||f(v)-f(u)||_2}{1-\epsilon}$
                 This yields $||v-u||_2 \le ||f(v)-f(u)||_2 +\frac
                 {\epsilon}{1-\epsilon} ||f(v)-f(u)||_2.$ Since $\epsilon \in
                 (0,1/2),$ the second inequality in this lemma follows.
                                  \qed
               \end{proof}          

                 Lemma \ref{lem: ineq2} immediately yields
                 the following lemma.
\begin{lemma}\label{lem: ineq3}
                  Assume the notation from Fact \ref{fact: achi03}.
                  Let $f(P)=\{ f(p_1),\dots,f(p_n)\}$,
                  $\ell'=\min \{ \ell , |f(P)|\},$
                  and let $T'$ be a set of $\ell'$ centers for the 
                  point set $f(P)$ in $\mathbb{R}^k.$
For $v' \in f(P),$ let $f^{-1}(v')=p_q,$ where $q=\min_i f(p_i)=v'.$
For any $\epsilon \in (0,1/2),$
the following inequalities hold w.h.p.:
                 $$(1-\epsilon)\max_{v'\in f(P)}
                 \min_{u'\in T'} ||v'-u'||_2\le
                 \max_{v\in P} \min_{u'\in T'} ||v-f^{-1}(u')||_2,$$
                 $$\max_{v\in P} \min_{u'\in T'} ||v-f^{-1}(u')||_2\le
                 (1+2\epsilon)\max_{v'\in f(P)}\min_{u'\in T'} ||v'-u'||_2.$$
               \end{lemma}
               
\begin{lemma}\label{lem: app}
  Assume the notation from Fact \ref{fact: achi03}.
  If $\epsilon \in (0,1/2)$ and the conservative subroutine $SR$ is
  $\alpha$-approximate then w.h.p.
  $DIMREDCENTER(\ell,P,\epsilon,SR)$ returns a
  $(1+\epsilon)(1+2\epsilon)\alpha$-approximate conservative solution to the
  $\ell$-center clustering problem for $P$ in the Euclidean
$\mathbb{R}^d.$
\end{lemma}
                    
\begin{proof}
Let $T_1$ be an optimal $\ell$-center
solution to the $\ell$-center clustering problem for $P\subset \mathbb{R}^d.$
Similarly, let $T_2$ be an optimal $\ell'$-center solution to
the $\ell'$-center clustering problem for $f(P)\subset \mathbb{R}^k.$
Next, let $r_1= \max_{v\in P} \min_{u\in T_1} ||v-u||_2$
and $r_2=\max_{v'\in f(P)} \min_{u'\in T_2} ||v'-u'||_2.$
By Lemma \ref{lem: ineq1},
 we have $(1-\epsilon)r_1\le r_2 \le (1+\epsilon)r_1.$
The procedure first computes an $\alpha$-approximate
$\ell'$-center solution $T'$ to the $\ell'$-center clustering
problem for $f(P)$ in $\mathbb{R}^k.$ It follows that
$\max_{v'\in f(P)} \min_{u'\in T'} ||v'-u'||_2\le (1+\epsilon)\alpha r_1$.
The procedure returns \\
$T=\{ f^{-1}(u')| u'\in T'\}$ extended
by $\ell -\ell'$ arbitrary points in $P\setminus T$
as an approximate
$\ell$-center solution to the $\ell$-center clustering problem
for the input point set $P\subset \mathbb{R}^d.$
The second inequality in Lemma~\ref{lem: ineq3} implies\\
$\max_{v\in P} \min_{u'\in T'} ||v-f^{-1}(u')||_2\le
(1+2\epsilon)\max_{v'\in f(P)}\min_{u'\in T'} ||v'-u'||_2.$
We conclude that $T$ yields  a $(1+2\epsilon)(1+\epsilon)\alpha$
approximation to the $\ell$-center clustering problem for $P.$
\qed
\end{proof}

\begin{lemma}\label{lem: timecenter1}
   All steps of  $DIMREDCENTER(\ell,P,\epsilon,SR)$ 
but for Step 4 can be implemented in
    $O((nd\log n)/\epsilon^2)$ time.
\end{lemma}
\begin{proof}
  Step 1 can be done in $O(nd)$ time. Step 2 takes $O(dk)\
  =\ O((d\log n)/\epsilon^2)$
  time. The preprocessing in Step 3 requires
$O(ndk)$, i.e., 
$O((nd\log n)/\epsilon^2)$ time.
Finally, Step 5 takes $O(nd)$ time.
\junk{Steps 5(a) and 5(c) take $O(n)$ time while Step 5(b) as
Step 4 can be done in $O((n\log n)/\delta^2)$ time.
Consequently, the whole Step 5 requires $O((n\ell \log n)/\delta^2 )$
time. Finally, Step 6 takes $O((n\log n)/\delta^2)$ time
similarly as Steps 4 and 5(b). It remains to observe that 
the overall running time of $CENTER(\ell,P,\delta)$
is dominated by those of Step 3 and Step 5.}
\qed
\end{proof}

\section{Fast $(2+\epsilon)$-approximation for $\ell$-center clustering}

In this section, we provide a faster randomized
$(2+\epsilon)$-approximation methods for the $\ell$-center clustering
and minimum-diameter $\ell$-clustering problems in the Euclidean
$\mathbb{R}^d$ and the Hamming space $\{0,1\}^d$. They are faster than
the known methods with approximation guarantee close to $2$ when
$d=\omega (\log n)$ and $\ell=\omega (\log n)$. Our method
in the Euclidean $\mathbb{R}^d$ is obtained
by plugging Gonzalez' method as the subroutine in the procedure
$DIMREDCENTER.$
\vfill
\newpage
\begin{theorem} \label{theo: 3plus}
   Let $P$ be a set of $n$ points $p_1,\dots,p_n \in \mathbb{R}^d,$
   $\ell$ an integer smaller than $n,$ and let $\epsilon\in (0,1/2).$
   The $\ell$-center clustering problem for $P$
   admits a conservative randomized approximation algorithm that
   w.h.p. provides 
   a $(2+\epsilon)$
   approximation of an optimal $\ell$-center clustering of $P$
   and runs in time\\
    $O(n\log n(d + \ell)/\epsilon^2).$
  Its slight modification yields a
   $(2+\epsilon )$ approximation to an $\ell$-clustering of $P$ with
   minimum cluster diameter w.h.p. and runs in the same asymptotic time.
 \end{theorem}
 \begin{proof}
   To prove the first part, we run $DIMREDCENTER(\ell,P, \epsilon/8, GO),$
   where $GO$ stands for the conservative Gonzalez' $2$-approximation
   algorithm for the $\ell$-center clustering problem.
   By Lemma \ref{lem: app},
   it yields a $(1+\epsilon/8)(1+2\epsilon/8)2\le$ $(2+\epsilon)$
   approximation of an optimal solution. Since Step 4 takes
   $O(n\ell k)=O((n\ell \log n)/ \epsilon^2)$ time, the whole
   call of the procedure can be implemented in time
   $O(n\log n(d + \ell)/\epsilon^2)$ by Lemma \ref{lem: timecenter1}.

   To prove the second part, we extend $DIMREDCENTER(\ell,P,\epsilon/8,GO)$
   slightly. Consider the set $T'$
   of $\ell'$ centers of $f(P)$ constructed by
   Gonzalez' algorithm in Step 4 of $DIMREDCENTER(\ell,P,\epsilon/8,GO).$
   In Step 5, additionally
   the $\ell'$ clusters $P_1,\dots,P_{\ell'}$ are
   formed by assigning each point $v\in P$ to the center in $T'$
   that is closest to $f(v).$ This extension takes
   $O(n\log n(d + \ell)/\epsilon^2)$ time.
   Let $s$ be a point in $f(P)\setminus T'$ that maximizes
   the distance from $T'.$ Let us denote this distance by $r'.$
   Note that for each point in $f(P)$ its distance to the closest
   center in $T'$ does not exceed $r'$ and importantly any two
   points in $T'\cup \{s\}$ are at least $r'$ apart by
   the furthest-point traversal done by Gonzalez' algorithm.
   It follows from Lemma \ref{lem: ineq1} that any two points
   in the set $f^{-1}(T')\cup \{f^{-1}(s)\}$
   are at least $(1-\epsilon/8)r'$ apart.
   Assume that $\ell'=\ell.$ Then the latter set
 has $\ell+1$ elements, at least two of them have to belong to the same
 cluster in any $\ell$-clustering of $P.$ Consequently, the maximum diameter
 of a cluster in any $\ell$-clustering of $P$ is at least
 $\frac{r'}{1+\epsilon/8}$.
 On the other hand, the diameter of any cluster $P_i$ is
 at most $2\frac {r'} {1-\epsilon/8}$ by Lemma \ref{lem: ineq1}. Now it is sufficient to observe
 that  $\frac {2(1+\epsilon/8)}{1-\epsilon/8}\le 2+\epsilon.$

To complete the proof observe
that we may assume w.l.o.g. $\ell=\ell'$ w.h.p. Simply,
we may assume w.l.o.g. that there are at least $\ell +1$
non-overlapping points in $P$ since otherwise the problem
admits a trivial solution. Let $t$ be a minimum distance
between a pair of $\ell +1$ non-overlapping points. W.h.p.
     each pair of $f$-images of these $\ell+1$ points is apart at least by
     $(1-\epsilon)t$ by Corollary~\ref{cor: newachi}.
  We conclude that $|f(P)|>\ell$ w.h.p.
\junk{
 To complete the proof we show
 that one can assume w.l.o.g. $\ell=\ell'$ w.h.p. by the following
 argumentation. Run Gonzalez' algorithm
     on the input set $P$ and let $q$ be a point in $P$
     maximizing the distance to the closest center provided
     by this algorithm. Denote the aforementioned
     distance by $t.$ It follows that any two points among the centers
     and the point $q$ have to be apart by at least
     $t$. We may assume w.l.o.g. that $t>0$ since otherwise the problem
     admits a trivial solution. But then consequently, w.h.p.
     each pair of $f$-images of these $\ell+1$ points is apart at least by
     $(1-\epsilon)t$ by Corollary \ref{cor: achi03}.
  We conclude that $|f(P)|>\ell$ w.h.p.}
 \qed  
\end{proof}
Let us recall the definitions of the $\ell$-center clustering and
minimum-diameter $\ell$-clustering problems in a Hamming space.

The $\ell$-center clustering problem in 
 the Hamming space $\{0,1\}^d$ is as follows: given a set $P$ of $n$ points
 in  $\{0,1\}^d$, find a set $T$ of $\ell$ points in  $\mathbb{R}^d$
 that minimize \\$\max_{v\in P} \min_{u\in T} \mathrm{ham}(v,u).$

The minimum-diameter $\ell$-clustering problem in 
 the Hamming space $\{ 0,1\}^d$ is as follows: given a set $P$ of $n$ points
 in  $\{0,1\}^d$, find a partition of $P$ into $\ell$ subsets
 $P_1,P_2,\dots, P_{\ell}$ that minimize
 $\max_{i\in [\ell]} \max_{v,u \in P_i} \mathrm{ham}(v,u).$

 To derive a result analogous to Theorem \ref{theo: 3plus} for
 Hamming spaces we cannot directly apply the general method
 for Euclidean spaces from Section 3 and instead use
 the following procedure.
 
\bigskip
\noindent\fbox{%
    \parbox{\textwidth}{%
    \noindent
        {\bf procedure} $DIMREDHAMCENTER(\ell, P, \epsilon )$
        \par
        \noindent
            {\em Input}: A positive integer $\ell,$
a set $P$ of points $p_1,\dots,p_n \in \{0,1\}^d,$ $n>\ell,$
and a real $\epsilon\in (0,\frac 12).$
            \par
            \noindent
            {\em Output}: A set $T\subset P$ of $\ell$
            centers of $P.$

                \begin{enumerate}
\item
                  Set $n$ to the number of input points
                 and $k$ to $O(\log n /\epsilon^2)$.
                 \item
                 Generate a random $k\times d$     matrix $R$ with entries
                 in $\{-1,1\}$,
                defining the function $f:\mathbb{R}^d \rightarrow 
                \mathbb{R}^k$ by $f(x)=Rx^{\top}$ 
                (see Fact \ref{fact: achi03}).
                \item
                  Compute the values of the function $f$  on each
                  point $p_i\in P,$
                  i.e.,
                  for $i \in [n],$
                  compute $Rp_i^{\top}.$ 
\item
Set $T$ to $\{ p_1 \}$, and for $j\in [n] \setminus \{ 1 \}$,
set $W_{1j}$  to $(||f(p_1)-f(p_j)||_2)^2$\\
(see Corollary  \ref{cor: achi03}).
\item
$\ell -2$ times iterate the following three steps:
\begin{enumerate}
\item
Find $p_m \in P\setminus T$ that maximizes $\min_{p_q\in T}
W_{qm}$ and extend $T$ to $T\cup \{p_m\}.$
\item
  For each $p_j\in P\setminus T,$ set $W_{mj}$
 to $(||f(p_m)-f(p_j)||_2)^2.$ 
\item
For each $p_j\in P\setminus T,$ update $\min_{p_i \in T} W_{ij}.$
\end{enumerate}
\item
Find $p_m \in P\setminus T$ that maximizes $\min_{p_q\in T}
W_{qm}$ and extend $T$ to $T\cup \{p_m\}.$
\item
Return $T.$
\end{enumerate}
}%
}
\vskip 5pt
By the specification of $W_{ij}$ in $DIMREDHAMCENTER(\ell, P, \epsilon )$
                and Corollary \ref{cor: achi03}, we obtain immediately
                the following lemma.
                
\begin{lemma}\label{lem: ineq}
  Assume the notation from
  $DIMREDHAMCENTER(\ell, P, \epsilon)$.
  For
  $i\in [n]$ and $j\in [n],$
the following inequalities  hold  w.h.p.:

$$W_{ij} \le (1+\epsilon)\mathrm{ham}(p_i,o_j),$$
$$(1-\epsilon)\mathrm{ham}(p_i,p_j)\le W_{ij}.$$
\end{lemma}

\begin{lemma}\label{lem: timecenter}
    $DIMREDHAMCENTER(\ell,P,\epsilon)$ runs in time
    $O(n\log n(d + \ell )/\epsilon^2)$.
\end{lemma}
\begin{proof}
  Step 1 can be done $O(nd)$ time. Step 2 takes $O(dk)\
  =\ O((d\log n)/\epsilon^2)$
  time. The preprocessing in Step 3 requires
$O(ndk)$, i.e., 
$O((nd\log n)/\epsilon^2)$ time.
Step 4 can be done in
 $O((n\log n )/\epsilon^2 )$ time.
Steps 5(a) and 5(c) take $O(n)$ time while Step 5(b) as
Step 4 can be done in $O((n\log n)/\epsilon^2)$ time.
Consequently, the whole Step 5 requires $O((n\ell \log n)/\epsilon^2 )$
time. Finally, Step 6 takes $O((n\log n)/\epsilon^2)$ time
similarly as Steps 4 and 5(b). It remains to observe that 
the overall running time of $DIMREDHAMCENTER(\ell,P,\epsilon)$
is dominated by those of Step 3 and Step 5.
\qed
\end{proof}
\vfill
\newpage
 \begin{theorem} \label{theo: 2plus}
   Let $P$ be a set of $n$ points $p_1,\dots,p_n \in \{0,1\}^d,$
   $\ell$ an integer smaller than $n,$ and let $\epsilon\in (0,1/2).$
   $DIMREDHAMCENTER(\ell,P,\epsilon/5)$
   provides w.h.p. a $(2+\epsilon)$-approximation
   of an optimal $\ell$-center clustering of $P$
   in the Hamming space $\{0,1\}^d$  in time
    $O(n\log n(d + \ell)/\epsilon^2).$
  Its slight modification yields a
   $(2+\epsilon) $ approximation to an $\ell$-clustering of $P$ with
   minimum cluster diameter w.h.p.
 \end{theorem}
 \begin{proof}
   Let $T$ be the set of
   $\ell$ centers output by $DIMREDHAMCENTER(\ell,P,\epsilon/5)$.
   Next, let $r=\max_{v\in P} \min_{u\in T} \mathrm{ham}(v,u)$,
   $r_w=\max_{p_i\in P}\min_{p_j\in T} W_{ij}$, and let $p_q$
 be a point for which the latter maximum is achieved.
 It follows from Lemma \ref{lem: ineq} that $r_w\ge r(1-\epsilon)$
  w.h.p.
 By the specification of the
 procedure\\ $DIMREDHAMCENTER$ and the definition of $p_q,$
 the set $T\cup \{p_q\}$ consists of $\ell+1$ points
 such that for any pair
 $p_v,\ p_u$ of points in this set $W_{vu}\ge r_w$ holds.
 Consequently, w.h.p. these $\ell+1$
 points are at the Hamming distance at least $r_w/(1+\epsilon/5)$
 apart by Lemma \ref{lem: ineq}.
 Let $T^*$ be an optimal set of $\ell$ centers of $P.$
 Two of these $\ell+1$ points in  $T\cup \{p_q\}$
 must have the same nearest center in $T^*$.
 It follows that the Hamming distance of at least one of these two points
 to its nearest center
 in $T^*$ is at least $\frac {r_w}{2(1+\epsilon/5)}$ w.h.p.
 Since $r_w\ge r(1-\epsilon/5)$ w.h.p. by Lemma \ref{lem: ineq},
 we infer that $\max_{p_i\in P} \min_{p_j\in T^*} \mathrm{ham}(p_i,p_j)$
 is at least  $r\frac {1-\epsilon/5}{2(1+\epsilon/5)}$ w.h.p.
 It remains to note that $2\frac {1+\epsilon/5}{1-\epsilon/5}\le 2+\epsilon$
 by $\epsilon \le 1/2$.
 This combined with Lemma \ref{lem: timecenter}
completes the proof of the first part.

To prove the second part,
we can slightly modify\\  $DIMREDHAMCENTER(\ell,P,\epsilon/5)$
so it returns a partition of $P$
into clusters $P_i$, $i\in [\ell],$ each consisting of all points
in $P$ for which the $i$-th center is closest in terms
of the approximate $W_{ij}$ distances.
\junk{Note that the maximum
diameter of a cluster in this partition is at most
$2r_w/(1-\epsilon/5)$ while the maximum diameter
of a cluster in any $\ell$-clustering of $P$
is at least $r_w/(1+\epsilon/5)$ by Lemma \ref{lem: ineq}.
We obtain a $2+\epsilon$ approximation by $2\frac{1-\epsilon/5}{1+\epsilon/5}\le 2+\epsilon.$}
To implement the modification, we let to
update in Step 5(c) not only $\min_{p_i \in T} W_{ij}$ but
also the current center $p_i$ minimizing  $\min_{p_i \in T} W_{ij}$.
This slight modification does not alter the asymptotic complexity
of  $DIMREDHAMCENTER(\ell,P,\epsilon/5)$.

The proof of the $2+\epsilon$ approximation guarantee in 
the second part follows by similar arguments as those in the first part.
Consider the $\ell +1$ points defined in the proof of the first part.
Recall that they are at least  $r_w/(1+\epsilon/5)$ apart w.h.p.
 Two of the $\ell+1$ points have also to belong to the same cluster
 in an $\ell$-clustering that minimizes the diameter. It follows
 that the minimum diameter is at least $r_w/(1+\epsilon/5)$ w.h.p.
 while the diameter of the clusters in the $\ell$-center
 clustering returned by $DIMREDHAMCENTER(\ell,P,\epsilon/5)$
 is at most  \\
 $2r_w/(1-\epsilon/5)$ w.h.p. by Lemma \ref{lem: ineq}.
 Consequently, it is larger by at most $\frac {2r_w/(1-\epsilon/5)}
 {r_w/(1+\epsilon/5)} \le 2 +\epsilon$ than the optimum w.h.p.
 \qed
\end{proof}

\section{Faster $O(\alpha)$-approximations for $\ell$-center clustering}

In this section, we shall plug recent fast approximation methods
for the $\ell$-center clustering problem in Euclidean spaces
in the procedure $DIMREDCENTER$ in order to decrease the heavy dependence
of their running times on the dimension
by minimally worsening their
approximation guarantees.
\junk{
\begin{fact}\label{fact: app1} \cite{EHS20}
  For $\alpha \ge 1,$ the $\ell$-center clustering problem
  for a set $P$ of $n$ points in the Euclidean $\mathbb{R}^d$
  admits $O(\alpha)$-approximation in $poly(d\log n)n^{1+1/\alpha^2}$
  time.
\end{fact}}

The following fact is a recent refinement of
an earlier result in \cite{EHS20}.

\begin{fact}\label{fact: app2} \cite{FJLNP25}$+$\cite{EHS20}
  For $\alpha \ge 1,$ the $\ell$-center clustering problem
  for a set $P$ of $n$ points in the Euclidean $\mathbb{R}^d$
  admits $O(\alpha)$-approximation in time\\ $poly(d\log n)
  (n+\ell^{1+1/\alpha^2}n^{O(1/\alpha^{2/3})})$, or alternatively, in
  time $poly(d\log n)n\ell^{1/\alpha^2}$. 
\end{fact}

Importantly, we may assume w.l.o.g.
that the $O(\alpha)$-approximation method in Fact \ref{fact: app2}
is conservative.
Simply, otherwise we can replace the current centers with closest
input points which at most doubles the distance of an
input point to the closest center.

By setting the subroutine $SR$ in $DIMREDCENTER(\ell,P,\epsilon,SR)$ to the
method of Fact \ref{fact: app2} and using Lemmata 
\ref{lem: app} and \ref{lem: timecenter1} with $\epsilon=\Omega(1)$,
we can reduce the dependence of the running time
on $d$ substantially.
\junk{
\begin{theorem}\label{theo: app1}
  For $\alpha \ge 1,$ the $\ell$-center clustering problem
  for a set $P$ of $n$ points in the Euclidean $\mathbb{R}^d$
  admits w.h.p. $O(\alpha)$-approximation
  in $\tilde{O}(nd/\epsilon^2+n^{1+1/\alpha^2})$
  time.
\end{theorem}

The next fact is a recent refinement of Fact \ref{fact: app1}.

\begin{fact}\label{fact: app2} \cite{FJLNP25}
  For $\alpha \ge 1,$ the $\ell$-center clustering problem
  for a set $P$ of $n$ points in the Euclidean $\mathbb{R}^d$
  admits $O(\alpha)$-approximation in time $poly(d\log n)
  (n+\ell^{1+1/\alpha^2}n^{O(1/\alpha^{2/3})})$, or alternatively, in
  time $poly(d\log n)n\ell^{1/\alpha^2}$. 
\end{fact}

Analogously, by plugging the methods of Fact \ref{fact: app2}
in \\ $DIMREDCENTER(\ell,P,\epsilon,SR)$ and using Lemmata 
\ref{lem: app} and \ref{lem: timecenter1} with
$\epsilon=\Omega(1)$,
we can reduce the dependence
of the running times of the methods on $d$ \\ substantially.}

\begin{theorem}\label{theo: app2}
  For $\alpha \ge 1,$ the $\ell$-center clustering problem
  for a set $P$ of $n$ points in the Euclidean $\mathbb{R}^d$
  admits w.h.p. $O(\alpha)$-approximation in time\\
  $\tilde{O}(nd/\epsilon^2 + \ell^{1+1/\alpha^2}n^{O(1/\alpha^{2/3})})$,
  or alternatively,
  in time $\tilde{O}(nd/\epsilon^2+n\ell^{1/\alpha^2})$.
\end{theorem}

\section{Fast $O(1)$-approximation
    for $\ell$-center clustering with outliers}

  In the variants of the $\ell$-center clustering and minimum-diameter
  $\ell$-clustering problems with outliers, a given number $z$ of
  input points can be discarded as outliers when trying to minimize
  the maximum distance to the nearest center or the maximum cluster
  diameter \cite{Char01,lot}.

  Charikar {\em et al.} were the first to provide a polynomial-time
    $O(1)$-approximation to the $\ell$-center clustering problem
    with outliers \cite{Char01}. They used a greedy method.
  
  \begin{fact}\label{fact: out}\cite{Char01}
    Given a set $P$ of $n$ points from an arbitrary metric, an integer
    $\ell \le n,$ and an integer $z,$ there is
    a polynomial-time $3$-approximation algorithm
    for the $\ell$-center clustering problem with $z$
    outliers in the metric.
  \end{fact}

  In the proof of Fact \ref{fact: out}, the authors
  assume that the considered optimal solution is conservative
  (i.e., the centers are input points)
  but they claim that one can remove this assumption in case
  of an Euclidean space keeping
  the approximation guarantee.
  
  The general method of dimension reduction given in
  Section 3 cannot be applied directly to the generalization
  of the $\ell$-center clustering problem to include outliers
  since several potential outliers may be mapped into a single
  point by the dimension reduction map. For this reason, we just
  modify the original greedy method from \cite{Char01} using
  the approximate interdistances based on the reduction instead
  of the real ones. In this way, we can significantly speed-up
  the greedy method at the cost of slightly increasing the approximation
  guarantee. In particular, we can replace an $O(n^2d)$
  component of the time complexity of the method by $\tilde{O}(n^2).$
  
  \bigskip
\noindent\fbox{%
    \parbox{\textwidth}{%
    \noindent
        {\bf procedure} $DIMREDCENOUT(\ell, P, \epsilon , z)$
        \par
        \noindent
            {\em Input}: A positive integers $\ell,z,$
a set $P$ of points $p_1,\dots,p_n \in \mathbb{R}^d,$ where $n>\ell,z,$
a real $\epsilon\in (0,\frac 12).$
            \par
            \noindent
            {\em Output}: A set $T$ of $\ell$ centers for
            a subset of $P$ of cardinality $n-z.$

                \begin{enumerate}
\item
                  Set $n$ to the number of input points
                 and $k$ to $O(\log n /\epsilon^2)$.
                 \item
                 Generate a random $k\times d$     matrix $R$ with entries
                 in $\{-1,1\}$,
                defining the function $f:\mathbb{R}^d \rightarrow 
                \mathbb{R}^k$ by $f(x)=Rx^{\top}$ 
                (see Fact \ref{fact: achi03}).
                \item
                  Compute the values of the function $f$  for each
                  point $p_i\in P,$
                  i.e.,
                  for $i=1,\dots,n,$
                  compute $Rp_i^{\top}.$
                 \item
                  For $i,j \in [n],$ compute $W'_{ij}=
                ||f(p_i)-f(p_j)||_2$ and set $W'$ to the
                matrix $(W'_{ij}).$
               \item
                 \begin{enumerate}
                   \item
                   Compute and sort the set
                   $B=\{ W'_{ij}|i,j\in [n]\} \cup \{ W'_{ij}(1+2\epsilon) |i,j\in [n]\}$.
                   \item 
               By binary search find the smallest $r$ in the sorted $B$
               such that $GREEDY(\ell,W',\epsilon,z,r)$ returns $YES$.
               \end{enumerate}
\item
  Set $T$ to the set of centers returned by the
  successful call of $GREEDY$
  with the smallest $r.$
\end{enumerate}
}%
}

\bigskip
\noindent\fbox{%
    \parbox{\textwidth}{%
    \noindent
        {\bf procedure} $GREEDY(\ell,W',\epsilon,z,r)$
        \par
        \noindent
        {\em Input}: The input parameters $\ell, P, \epsilon, z,$
        $W'=(W'_{ij})$ are specified as
        in $DIMREDCENOUT(\ell, P, \epsilon , z)$, $r$ is a positive
        real number.
            \par
            \noindent
            {\em Output}: YES if there is an $\ell$-center
            clustering of a $(n-z)$-point subset of $P$
            such that the maximum $\ell_2$ distance
            of a point in the subset to its nearest center
            does not exceed $3(1+\epsilon)r$ otherwise NO.
                \begin{enumerate}
   \item
  For $i\in [n],$ compute the set $G_i$ of points $p_j\in P$
  s.t. $W'_{ij}\le r(1+\epsilon)$ and the set $E_i$
  of points $p_j\in P$ s.t.$W'_{ij}\le 3r(1+\epsilon)$. Set $S$ to $[n].$
\item
  $\ell$ times iterate the following block:
  \begin{enumerate}
    \item
      Select a set $G_j$ of largest cardinality
      among (the not yet selected) sets $G_i,i\in S$.
      Select also $E_j.$ Set $S$ to $S\setminus \{j\}.$
  \item
    Mark the points in $E_j$ and remove
    the newly marked points from all (the not yet selected)
    sets
    $G_i, E_i,$ for $i\in S.$
  \end{enumerate}

\item
  If the total number of marked points is at least
  $n-z,$ i.e., $|\bigcup_{i=1}^{\ell}E_i|\ge n-z,$
  then return YES along with $E_i, i \in [\ell]$ else output NO.
  \end{enumerate}
}%
}

\begin{lemma}\label{lem: greedy}
  If there is a set
  $T$  of $\ell$ centers in
  $P \subset \mathbb{R}^d$ such that
  for at least $n-z$ points $p \in P,$ $\min_{t\in T}||t-p||_2\le r$
  then $GREEDY(\ell,P,W',\epsilon, z,r)$ returns
YES w.h.p.
\end{lemma}
\begin{proof}
  We may assume w.l.o.g. that $p
  _1,\dots,p_{\ell}$
  are the centers in the sets $G_i,$ and their extensions
  $E_i,$ $i\in [\ell],$ produced by
$GREEDY(\ell,P,W',\epsilon, z,r)$.
Let $T=\{ t_1,\dots,t_{\ell}\}$ and for $i\in [\ell],$
let $T_i=\{p\in P|\ ||t_i-p||_2\le r\}.$
Following \cite{Char01}, to prove the lemma it is sufficient
to show by induction on $i\in [\ell]$ that one can
order the sets $T_i$ so that $E_1\cup E_2 \dots E_i$
cover at least as many points in $P$ as $T_1\cup T_2\dots T_{i}$ w.h.p.
The proof is by assigning to each point in the latter set union
a distinct point in the former set union. 

By the inductive hypothesis, we may assume that the sets
$T_1, T_2,\dots,T_{i-1}$ and the assignment of a distinct 
point in $E_1\cup E_2 \dots E_{i-1}$ to each point in
$T_1\cup T_2 \dots T_{i-1}$ 
have been determined. Suppose that the set $G_i$ intersects
some remaining set $T_j.$ Then, $E_i$ covers all points
in $T_j$ not covered by $E_1\cup E_2\dots E_{i-1}$ w.h.p. by
Corollary \ref{cor: newachi}. Therefore,  we can assign to each point
in $T_j$ not covered by $E_1\cup E_2\dots E_{i-1}$
the point itself. Otherwise, by the greedy choice of $G_i$ and
Corollary \ref{cor: newachi},
$G_i$ and consequently also $E_i$ contain at least as many points
outside $E_1\cup E_2\dots E_{i-1}$ as $T_j$. Therefore,
we can assign to each point in $T_j$ not covered by
$E_1\cup E_2\dots E_{i-1}$
a distinct point in $E_i.$ Consequently, we can
further rearrange the order of the
remaining sets $T_q, i\le q \le \ell,$ so
$T_j$ becomes $T_i.$
Importantly, no point can be doubly assigned since
in the $i$-th iteration of the $GREEDY$ procedure,
the updated sets $E_i$ are disjoint from
the sets $E_q, q < i.$
\qed
\end{proof}

\begin{lemma}\label{lem: timegreedy}
  Assume the notation from Lemma \ref{lem: greedy}. 
  $GREEDY(\ell,P,W',\epsilon, z,r)$ can be implemented
  in $\tilde{O}(n^2\ell)$ time. 
\end{lemma}
\begin{proof}
  In Step 1, we compute representations
  of the sets $G_i,E_i$  in the form of dictionaries
  keeping the indices of the points belonging
  to these sets.
  \junk{In order to compute the representations,
  for each pair of points $p_i,\ p_j\in P,$
  we need to compute $||f(p_i)-f(p_j)||_2$ which takes
  $O(n^2k)=O((n^2\log n)/\epsilon^2)$ time. }
The dictionaries can be formed in $\tilde{O}(n^2)$ time.
  A single iteration of
the block in Step~2 takes $\tilde{O}(n^2)$
time. Hence, the whole Step 2 takes $\tilde{O}(n^2\ell)$
time.
Finally, Step 3 can be done in $O(n)$ time.
\qed
\end{proof}

  \begin{lemma}\label{lem: timeout}
    $DIMREDCENOUT(\ell, P, \epsilon , z)$ can be implemented
    in time\\ $\tilde{O}(n^2(\epsilon^{-2}+ \ell))$.
  \end{lemma}
  \begin{proof}
    Steps 1-3 are analogous to those in
    $DIMREDCENTER(\ell,P,\epsilon,SR)$ and hence
    they can be implemented in  time
    $O((nd\log n)/\epsilon^2)$ by Lemma \ref{lem: timecenter1}.
    Step~4 takes $O(n^2k)=O((n^2\log n)/\epsilon^2)$
    time. Step 5(a) requires $\tilde{O}(n^2)$ time.
    By Lemma~\ref{lem: timegreedy}, Step 5(b)
    can be done in $\tilde{O}(n^2\ell).$
    Finally, Step 6 can be completed in $O(nd)$ time.
    \qed
  \end{proof}
  
  \begin{theorem}\label{theo: out}
    Let $P$ be a set of $n$ points in the Euclidean $\mathbb{R}^d,$
    and let $z$ be a nonnegative integer smaller than $n.$
    Suppose that there is a set $T \subset P$
    of $\ell$ centers that is a solution to the $\ell$-center
    clustering problem for $P$
    with $z$ outliers in the Euclidean space, where the maximum
    distance of a non-outlier to its closest center in $T$ is at most
    $r.$ Let $\epsilon \in (0,\frac 12).$
    In time $\tilde{O}(n^2(\epsilon^{-2}+ \ell))$, one can construct
    a conservative  solution to
    the $\ell$-center clustering problem for $P$ with $z$ outliers
    in the Euclidean space such that the maximum
    distance of a non-outlier to its closest center is at most
    $(3+\epsilon)r$ w.h.p.
  \end{theorem}
  \begin{proof}
    Let us run $DIMREDCENOUT(\ell,P,\delta,z)$,
    where $\delta$ is a fraction of $\epsilon$ to be specified
    later. Note that $r=||p_i-p_j||_2$ for some
    $i,j\in [n].$ Recall also that both $W'_{ij}$ and
    $W'_{ij}(1+\delta)$ are considered in the binary search in
    Step 5 of $DIMREDCENOUT(\ell,P,\delta,z)$. If $W'_{ij}\le r$
    then $r \le (1+2\delta)W'_{ij}\le (1+2\delta)r$ by Lemma \ref{lem: ineq2}.
    Otherwise, we have $r\le W'_{ij}\le (1+\delta)r$ by Corollary
    \ref{cor: newachi}.
   We infer that in
    the binary search, a value $r'$ between $r$ and $(1+2\delta)r$
    is considered. It follows from Lemma \ref{lem: greedy}, that
    $DIMREDCENOUT(\ell,P,\delta,z)$ produces a conservative solution
    where the maximum distance of non-outlier to its closest center
    is at most $(3+\delta)(1+2\delta)r$. This is at most
    $(3+\epsilon)r$ when $\delta$ is set to $\frac {\epsilon}{8} .$

    $DIMREDCENOUT(\ell,P,\epsilon/8 ,z)$ can be implemented in
    time $\tilde{O}(n^2(\epsilon^{-2}+ \ell))$ by Lemma \ref{lem: timeout}.
    \qed
  \end{proof}

  Since a solution to an instance of the $\ell$-center clustering
  problem with outliers can be easily transformed to a conservative
  one by at most doubling the distances of non-outlier points to their
  closest centers, we obtain the following corollary.
  
  \begin{corollary}
 Let $P$ be a set of $n$ points in the Euclidean $\mathbb{R}^d,$
 and let $z$ be a nonnegative integer smaller than $n.$
 For an arbitrary $\epsilon \in (0,\frac 12),$
 the $\ell$-center clustering
 problem for $P$ with $z$ outliers
 in the Euclidean space admits w.h.p. a $(6+\epsilon)$-approximation
 in time $\tilde{O}(n^2(\epsilon^{-2}+ \ell))$.
    \end{corollary}

    By slightly modifying the procedures $DIMREDCENOUT$
    and $GREEDY$, we can obtain
    analogous theorem and corollary for $\ell$-center clustering
    with outliers in Hamming spaces. Simply,
    in the modified procedures, we use as the approximate
    distance between points $v$ and $u$ the squared distance
    $||f(v)-f(u)||_2$ instead of  $||f(v)-f(u)||_2$. 
    The asymptotic running times of the modified
    procedures are the same as those of the original ones.
    The proof of the approximation guarantee is analogous to that
    in the Euclidean case. Instead of Corollary \ref{cor: newachi}
    we use Corollary \ref{cor: achi03} and instead of Lemma
    \ref{lem: ineq2} its Hamming equivalent. See Appendix A for details.
    
  \begin{theorem}\label{theo: hamout}
    Let $P$ be a set of $n$ points in the Hamming space
    $\{0,1\}^d,$
    and let $z$ be a nonnegative integer smaller than $n.$
    Suppose that there is a set $T \subset P$
    of $\ell$ centers that is a solution to the $\ell$-center
    clustering problem for $P$
    with $z$ outliers in the Hamming space
    $\{0,1\}^d,$ where the maximum
    Hamming distance of a non-outlier to its closest center in $T$ is at most
    $r.$ Let $\epsilon \in (0,\frac 12).$
    In time\\ $\tilde{O}(n^2(\epsilon^{-2}+ \ell))$, one can construct
    a conservative  solution to
    the $\ell$-center clustering problem for $P$
    with $z$ outliers in the Hamming space
    $\{0,1\}^d$
    such that the maximum Hamming
    distance of a non-outlier to its closest center is at most
    $(3+\epsilon)r$ w.h.p.
  \end{theorem}

  \begin{corollary}\label{cor: hamout}
 Let $P$ be a set of $n$ points in the Hamming space
    $\{0,1\}^d,$
 and let $z$ be a nonnegative integer smaller than $n.$
 For an arbitrary $\epsilon \in (0,\frac 12),$ the $\ell$-center clustering
 problem for $P$
 with $z$ outliers in the Hamming space
    $\{0,1\}^d$ admits w.h.p. a $(6+\epsilon)$-approximation
 in time $\tilde{O}(n^2(\epsilon^{-2}+ \ell))$.
\end{corollary}

\section{Final remarks}

Relatively recently, Jiang {\em et al.} have shown in \cite{JKS24}
that for some problems, including the $\ell$-center clustering
problem in Euclidean spaces, there is a randomized dimension
reduction even to a sublogarithmic-size subspace.
Such a reduction preserves the optimal value of the solution
to the problem within a factor related to the size of the subspace.
Following \cite{JKS24}, their main result can be stated informally
as follows.

\begin{fact}\label{jks}\cite{JKS24}
  For every $d,\alpha, \ell, n,$ where $\ell \le n,$ there
  is a random linear map $g: \mathbb{R}^d \rightarrow \mathbb{R}^t,$
  where $t=O(\frac{\log n}{\alpha^2}+ \log \ell),$ such that for every
  set $P \subset \mathbb{R}^d$ of $n$ points, with high probability,
  $g$ preserves the value of optimal solution to the
  $\ell$-center clustering problem within $O(\alpha)$ factor.
  \end{fact}

  We could replace the variant of JL randomized dimension reduction
  from \cite{achi03} in the general method presented in Section 3
  by the enhanced randomized
  dimension reduction from \cite{JKS24}. Roughly, this
  would lead to decreasing the asymptotic
  time complexity by $O(\alpha^2)$ at the cost of increasing the
  approximation guarantee by $O(\alpha).$ As the authors of \cite{JKS24}
  note, when one is interested in $O(1)$-approximation and hence $\alpha=O(1),$
  then $\frac {\log n}{\alpha^2}=\Omega (\log n)$ and their method
  does not yield any improvement of the asymptotic running time.
  On the other hand, the randomized dimension reduction from
  \cite{JKS24} can be advantageous when for instance algorithms
  of exponential dependence on the dimension size are
  applied in the target subspace.
\junk{
  Similarly as in \cite{JKS24}, we could also extend our results to
  include some known variants of the $\ell$-center clustering and
  minimum-diameter $\ell$-clustering problems. For instance, in the
  variants with outliers, a given number $z$ of input points can be
  discarded as outliers when trying to minimize the maximum distance
  to the nearest center or the maximum cluster diameter
  \cite{Char01,lot}.
}
  
     \small
  \bibliographystyle{abbrv}                    
\bibliography{Voronoi4}
\vfill
\newpage
  \appendix
  \section{$O(1)$-approximation
    for $\ell$-center clustering with outliers in Hamming spaces}
  Our Hamming versions of the procedures $DIMREDCENOUT$ and $GREEDY$
  are as follows.

  \bigskip
\noindent\fbox{%
    \parbox{\textwidth}{%
    \noindent
        {\bf procedure} $HAMDIMREDCENOUT(\ell, P, \epsilon , z)$
        \par
        \noindent
            {\em Input}: A positive integers $\ell,z,$
a set $P$ of points $p_1,\dots,p_n \in \{0,1\}^d,$ where $n>\ell,z,$
a real $\epsilon\in (0,\frac 12).$
            \par
            \noindent
            {\em Output}: A set $T$ of $\ell$ centers for
            a subset of $P$ of cardinality $n-z.$

                \begin{enumerate}
\item
                  Set $n$ to the number of input points
                 and $k$ to $O(\log n /\epsilon^2)$.
                 \item
                 Generate a random $k\times d$     matrix $R$ with entries
                 in $\{-1,1\}$,
                defining the function $f:\mathbb{R}^d \rightarrow 
                \mathbb{R}^k$ by $f(x)=Rx^{\top}$ 
                (see Fact \ref{fact: achi03}).
                \item
                  Compute the values of the function $f$  for each
                  point $p_i\in P,$
                  i.e.,
                  for $i=1,\dots,n,$
                  compute $Rp_i^{\top}.$
                 \item
                  For $i,j \in [n],$ compute $W''_{ij}=
                (||f(p_i)-f(p_j)||_2)^2$ and set $W''$ to the
                matrix $(W''_{ij}).$
               \item
                 \begin{enumerate}
                   \item
                   Compute and sort the set
                   $B'=\{ W''_{ij}|i,j\in [n]\} \cup \{ W''_{ij}(1+2\epsilon) |i,j\in [n]\}$.
                   \item 
               By binary search find the smallest $r$ in the sorted $B'$
               such that $HAMGREEDY(\ell,W'',\epsilon,z,r)$ returns $YES$.
               \end{enumerate}
\item
  Set $T$ to the set of centers returned by the
  successful call of $HAMGREEDY$
  with the smallest $r.$
\end{enumerate}
}%
}

\bigskip
\noindent\fbox{%
    \parbox{\textwidth}{%
    \noindent
        {\bf procedure} $HAMGREEDY(\ell,W'',\epsilon,z,r)$
        \par
        \noindent
        {\em Input}: The input parameters $\ell, P, \epsilon, z,$
        $W''=(W''_{ij})$ are specified as
        in $HAMDIMREDCENOUT(\ell, P, \epsilon , z)$, $r$ is a positive
        real number.
            \par
            \noindent
            {\em Output}: YES if there is an $\ell$-center
            clustering of a $(n-z)$-point subset of $P$
            such that the maximum Hamming distance
            of a point in the subset to its nearest center
            does not exceed $3(1+\epsilon)r$ otherwise NO.
                \begin{enumerate}
   \item
  For $i\in [n],$ compute the set $G'_i$ of points $p_j\in P$
  s.t.  $W''_{ij}\le r(1+\epsilon)$ and the set $E'_i$
  of points $p_j\in P$ s.t.$W''_{ij}\le 3r(1+\epsilon)$. Set $S$ to $[n].$
\item
  $\ell$ times iterate the following block:
  \begin{enumerate}
    \item
      Select a set $G'_j$ of largest cardinality
      among (the not yet selected) sets $G'_i,i\in S$.
      Select also $E'_j.$ Set $S$ to $S\setminus \{j\}.$
  \item
    Mark the points in $E'_j$ and remove
    the newly marked points from all (the not yet selected)
    sets
    $G'_i, E'_i,$ for $i\in S.$
  \end{enumerate}

\item
  If the total number of marked points is at least
  $n-z,$ i.e., $|\bigcup_{i=1}^{\ell}E'_i|\ge n-z,$
  then return YES along with $E'_i, i \in [\ell]$ else output NO.
  \end{enumerate}
}%
}

\begin{lemma}\label{lem: hamgreedy}
  If there is a set
  $T$  of $\ell$ centers in
  $P \subset \{0,1\}^d$ such that
  for at least $n-z$ points $p \in P,$ $\min_{t\in T}(||t-p||_2)^2\le r$
  then $HAMGREEDY(\ell,P,W',\epsilon, z,r)$ returns
YES w.h.p.
\end{lemma}
\begin{proof}
  We may assume w.l.o.g. that $p
  _1,\dots,p_{\ell}$
  are the centers in the sets $G'_i,$ and their extensions
  $E'_i,$ $i\in [\ell],$ produced by
$HAMGREEDY(\ell,P,W'',\epsilon, z,r)$.
Let $T=\{ t_1,\dots,t_{\ell}\}$ and for $i\in [\ell],$
let $T_i=\{p\in P|\ (||t_i-p||_2)^2\le r\}.$
As in the Euclidean case, to prove the lemma it is sufficient
to show by induction on $i\in [\ell]$ that one can
order the sets $T_i$ so that $E'_1\cup E'_2 \dots E'_i$
cover at least as many points in $P$ as $T_1\cup T_2\dots T_{i}$ w.h.p.
The proof is by assigning to each point in the latter set union
a distinct point in the former set union. 

By the inductive hypothesis, we may assume that the sets
$T_1, T_2,\dots,T_{i-1}$ and the assignment of a distinct 
point in $E'_1\cup E'_2 \dots E'_{i-1}$ to each point in
$T_1\cup T_2 \dots T_{i-1}$ 
have been determined. Suppose that the set $G'_i$ intersects
some remaining set $T_j.$ Then, $E'_i$ covers all points
in $T_j$ not covered by $E'_1\cup E'_2\dots E'_{i-1}$ w.h.p. by
Corollary \ref{cor: achi03i}. Therefore,  we can assign to each point
in $T_j$ not covered by $E'_1\cup E'_2\dots E'_{i-1}$
the point itself. Otherwise, by the greedy choice of $G'_i$ and
Corollary \ref{cor: achi03},
$G'_i$ and consequently also $E'_i$ contain at least as many points
outside $E'_1\cup E'_2\dots E'_{i-1}$ as $T_j$. Therefore,
we can assign to each point in $T_j$ not covered by
$E'_1\cup E'_2\dots E'_{i-1}$
a distinct point in $E'_i.$ Consequently, we can
further rearrange the order of the
remaining sets $T_q, i\le q \le \ell$ so
$T_j$ becomes $T_i.$
Importantly, no point can be doubly assigned since
in the $i$-th iteration of the $HAMGREEDY$ procedure,
the updated sets $E'_i$ are disjoint from
the sets $E'_q, q < i.$
\qed
\end{proof}

\begin{lemma}\label{lem: timehamgreedy}
  Assume the notation from Lemma \ref{lem: hamgreedy}. 
  $HAMGREEDY(\ell,P,W'',\epsilon, z,r)$ can be implemented
  in $\tilde{O}(n^2\ell)$ time. 
\end{lemma}
\begin{proof} is analogous to that of Lemma \ref{lem: timegreedy}.
 \junk{ In Step 1, we compute representations
  of the sets $G'_i,E_i$  in the form of dictionaries
  keeping the indices of the points belonging
  to these sets.
  \junk{In order to compute the representations,
  for each pair of points $p_i,\ p_j\in P,$
  we need to compute $||f(p_i)-f(p_j)||_2$ which takes
  $O(n^2k)=O((n^2\log n)/\epsilon^2)$ time. }
The dictionaries can be formed in $\tilde{O}(n^2)$ time.
  A single iteration of
the block in Step~2 takes $\tilde{O}(n^2)$
time. Hence, the whole Step 2 takes $\tilde{O}(n^2\ell)$
time.
Finally, Step 3 can be done in $O(n)$ time.}
\qed
\end{proof}

  \begin{lemma}\label{lem: timehamout}
    $DIMREDCENOUT(\ell, P, \epsilon , z)$ can be implemented
    in time\\ $\tilde{O}(n^2(\epsilon^{-2}+ \ell))$.
  \end{lemma}
  \begin{proof} is analogous to that of Lemma \ref{lem: timeout}
  \junk{  Steps 1-3 are analogous to those in
    $DIMREDCENTER(\ell,P,\epsilon,SR)$ and hence
    they can be implemented in  time
    $O((nd\log n)/\epsilon^2)$ by Lemma \ref{lem: timecenter1}.
    Step 4 takes $O(n^2k)=O((n^2\log n)/\epsilon^2)$
    time. Step 5(a) requires $\tilde{O}(n^2)$ time.
    By Lemma~\ref{lem: timehamgreedy}, Step 5(b)
    can be done in $\tilde{O}(n^2\ell)$ time.
    Finally, Step 6 can be completed in $O(nd)$ time.}
    \qed
  \end{proof}
We need also a Hamming equivalent of Lemma \ref{lem: ineq2}.
  
  \begin{lemma}\label{lem: hamineq2}
                  Assume the notation from Fact \ref{fact: achi03}.
                  Suppose that $\epsilon \in (0,1/2)$
                  and $P\subset \{0,1\}^d\subset \mathbb{R}^d.$
                  Then, 
                  for all $v,u\in P,$
                  the following inequalities hold w.h.p.:
$$(1-\epsilon) (||f(v)-f(u)||_2)^2 \le \mathrm{ham}(v,u),$$ 
$$\mathrm{ham}(v,u)\le (1+2\epsilon)(||f(v)-f(u)||_2)^2.$$
               \end{lemma}
               \begin{proof}

                 By the right-hand
                 inequality in Corollary \ref{cor: achi03}, we
                 obtain for any $v, u \in P,$ and $\epsilon \in
                 (0,1/2),$ $\frac {(||f(v)-f(u)||_2)^2}{1+\epsilon} \le
                 \mathrm{ham}(v,u).$
                 It follows that $\mathrm{ham}(v,u)- \frac
                 {\epsilon}{1+\epsilon} (||f(v)-f(u)||_2)^2 \le
                 \mathrm{ham}(v,u).$ Consequently, we obtain the first
                 inequality in this lemma.

                 Similarly, by the left-hand inequality
                 in Corollary \ref{cor: achi03}, we
                 infer that for any $v, u \in P,$ and $\epsilon \in
                 (0,1/2),$ $\mathrm{ham}(v,u)
                 \le \frac {(||f(v)-f(u)||_2)^2}{1-\epsilon}$.
                 This yields $\mathrm{ham}(v,u) \le$\\ $(||f(v)-f(u)||_2)^2 +\frac
                 {\epsilon}{1-\epsilon} (||f(v)-f(u)||_2)^2.$
                 Since $\epsilon \in
                 (0,1/2),$ the second inequality in this lemma follows.
                                  \qed
                                \end{proof}
  \junk{                              
  \begin{theorem}\label{theo: hamout}
    Let $P$ be a set of $n$ points in the Hamming space
    $\{0,1\}^d,$
    and let $z$ be a nonnegative integer smaller than $n.$
    Suppose that there is a set $T \subset P$
    of $\ell$ centers that is a solution to the $\ell$-center
    clustering problem for $P$
    with $z$ outliers in the Hamming space
    $\{0,1\}^d,$ where the maximum
    Hamming distance of a non-outlier to its closest center in $T$ is at most
    $r.$ Let $\epsilon > 0.$
    In time $\tilde{O}(n^2(\epsilon^{-2}+ \ell))$, one can construct
    a conservative  solution to
    the $\ell$-center clustering problem for $P$
    with $z$ outliers in the Hamming space
    $\{0,1\}^d$
    such that the maximum Hamming
    distance of a non-outlier to its closest center is at most
    $(3+\epsilon)r$ w.h.p.
  \end{theorem}}
  \begin{proof} {\em of Theorem \ref{theo: hamout}}.
    Let us run $HAMDIMREDCENOUT(\ell,P,\delta,z)$,
    where $\delta$ is set to $\frac {\epsilon}{8} .$
    Note that $r=\mathrm{ham}(p_i,p_j)=(||p_i-p_j||_2)^2$ for some
    $i,j\in [n].$ Recall also that both $W''_{ij}$ and
    $W''_{ij}(1+\delta)$ are considered in the binary search in
    Step 5 of $HAMDIMREDCENOUT(\ell,P,\delta,z)$. If $W''_{ij}\le r$
    then $r \le (1+2\delta)W''_{ij}\le (1+2\delta)r$
    by Lemma \ref{lem: hamineq2}.
    Otherwise, we have $r\le W''_{ij}\le (1+\delta)r$ by Corollary 2.
   We infer that in
    the binary search, a value $r'$ between $r$ and $(1+2\delta)r$
    is considered. It follows from Lemma~\ref{lem: hamgreedy}, that
    $HAMDIMREDCENOUT(\ell,P,\delta,z)$ produces a conservative solution
    where the maximum distance of non-outlier toots closest center
    is at most $(3+\delta)(1+2\delta)r$. This is at most
    $(3+\epsilon)r$ by $\delta=\frac {\epsilon}{8} .$

    $HAMDIMREDCENOUT(\ell,P,\epsilon/8 ,z)$ can be implemented in
    time $\tilde{O}(n^2(\epsilon^{-2}+ \ell))$ by Lemma \ref{lem: timehamout}.
    \qed
  \end{proof}
\end{document}